\documentclass[aps,twocolumn,showpacs,amsmath,amssymb]{revtex4}
\usepackage{amsthm}
\usepackage{latexsym}
\usepackage{amsfonts}
\usepackage{bbm,dsfont}
\usepackage{graphicx}

\newtheorem{proposition}{Proposition}

\newcommand{\h}[1]{\mathcal{#1}}
\newcommand{\R}{\mathbb{R}}
\newcommand{\C}{\mathbb{C}}
\newcommand{\hil}{\mathcal{H}}

\newcommand{\lh}{\mathcal{L(H)}}
\newcommand{\sfq}{\mathsf{Q}}

\newcommand{\E}{\mathsf{E}}

\newcommand{\G}{\mathsf{G}}

\newcommand{\tr}[1]{\mathrm{tr}\left[ {#1} \right]}

\begin{document}

\title{Tilted phase space measurements}

\author{Pekka Lahti}
\email{pekka.lahti@utu.fi}
\affiliation{Turku Centre for Quantum Physics, Department of Physics and Astronomy, University of Turku, FI-20014 Turku, Finland}

\author{Jussi Schultz}
\email{jussi.schultz@utu.fi}
\affiliation{Turku Centre for Quantum Physics, Department of Physics and Astronomy, University of Turku, FI-20014 Turku, Finland}

\date{\today}
\begin{abstract}
We show that the phase shift of $\frac{\pi}{2}$ is crucial for the phase space translation covariance of the measured high-amplitude limit observable in eight-port homodyne detection. However, for an arbitrary phase shift $\theta$ we construct explicitly a different nonequivalent projective representation of $\R^2$ such that the observable is covariant with respect to this representation. As a result we are able to determine  the measured observable for an arbitrary parameter field and phase shift. Geometrically the change in the phase shift corresponds to the tilting of one axis in the phase space of the system.
\end{abstract}

\maketitle

Covariant phase space observables are an invaluable tool when studying many fundamental questions within the quantum theory. On one hand, the measurement of such an observable constitutes an approximate joint measurement of the position and momentum of a quantum system, thus allowing us to gain deeper insight into the full content of Heisenberg's uncertainty principle. These joint measurements are even optimal in some sense, since to any approximate joint observable for position and momentum there exists a covariant one which serves as a better approximation \cite{Werner2}. On the other hand, these observables can be used for the purpose of continuous variable quantum tomography. Indeed, a large class of covariant phase space observables are such that the measurement outcome statistics determine the state of the system uniquely \cite{Ali}. This is a fact whose significance increases alongside with the development of quantum information technology, where the possibility of determining the quantum state of a system is often of great significance.

The measurement of any covariant phase space observable can be realized experimentally via eight-port homodyne detection provided that a suitable parameter field is fed into one of the four input ports \cite{Kiukas2}. Therefore the detailed study of this particular scheme has strong physical motivation. The crucial ingredient responsible for the covariance is the phase shifter which is always assumed to provide a phase shift of $\frac{\pi}{2}$. In fact, we show that any deviation from this presumed value destroys the covariance of the observable. Since the structure of the observable depends so strongly on the covariance, it is not  justified to make any {\em a priori} assumptions on its properties when the phase shift is arbitrary.

In this Letter we give an explicit and rigorous construction of the measured observable for an arbitrary phase shift. We show that for each phase shift, there is a corresponding projective representation of $\R^2$ such that the observable is covariant with respect to it. In this way, we obtain explicit forms for a large class of nonequivalent representations and show that they have clear physical meanings, since the corresponding observables serve as approximate joint observables for pairs of rotated quadratures. The change of representations can also be interpreted geometrically; it corresponds to a tilting of one axis in the phase space of the system (see Figure \ref{axis}).

\begin{figure}
\includegraphics[width=8cm]{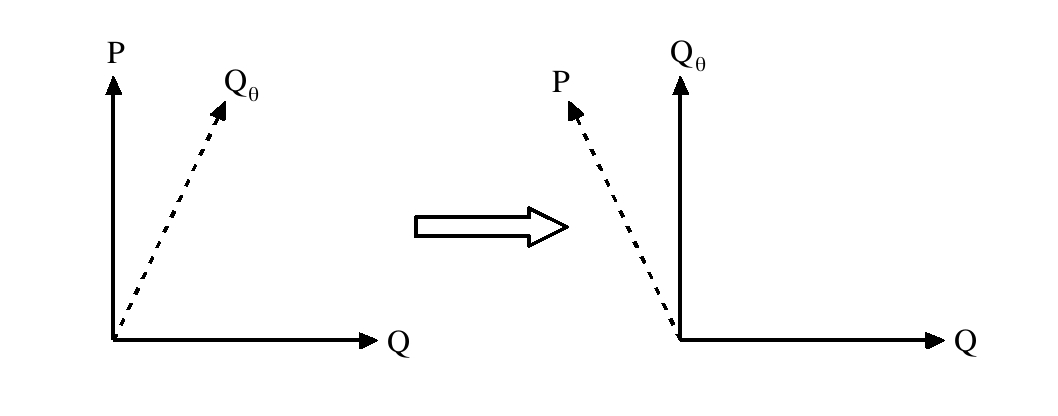}
\caption{Tilting of the phase space caused by a change in the phase shift.}\label{axis}
\end{figure}

Let us briefly recall the mathematical framework of our study. Let $\hil$ be the Hilbert space associated with a quantum system, and let $\lh$ denote the set of bounded operators acting on $\hil$. The states of the system are represented by positive operators with unit trace, and the observables are represented by normalized positive operator measures $\E :\h B(\R^n)\rightarrow \h L(\hil)$, where $\h B(\R^n)$ stands for the Borel $\sigma$-algebra of subsets of the measurement outcome space $\R^n$. For a system in a state $\rho$, the measurement outcome statistics of an observable $\mathsf{E}$ is given by the probability measure $\mathsf{E}_\rho :\h B (\R^n)\rightarrow [0,1]$, $\mathsf{E}_\rho (X) =\textrm{tr} [\rho \mathsf{E} (X)]$. Any two observables $\mathsf{E}$ and $\mathsf{F}$ are said to be {\em informationally equivalent} if they distinguish exactly the same states, that is, if $\mathsf{E}_{\rho_1} =\mathsf{E}_{\rho_2} $ if and only if $\mathsf{F}_{\rho_1} =\mathsf{F}_{\rho_2} $. If the measurement outcome statistics of $\mathsf{E}$ determine the state uniquely, then $\mathsf{E}$ is {\em informationally complete} \cite{Prugovecki}.

The eight-port homodyne detector consists of four input modes with Hilbert spaces $\hil_j$, $j=1,2,3,4$, four lossless $50:50$ beam splitters, a phase shifter and four photon detectors (see Figure \ref{eightport}). We fix the number basis $\{ \vert n\rangle \vert n=0,1,2,\ldots \}$ for each mode so that, in particular, the coherent states $\{ \vert z\rangle \vert z\in\C\}$  are defined by the expression
$$
\vert z\rangle = e^{-\frac{\vert z\vert^2}{2}} \sum_{n=0}^\infty \frac{z^n}{\sqrt{n!}} \vert n\rangle.
$$
The beam splitters are modelled by the unitary operators $U_{ij}$, determined by their action on the coherent states:
$$
U_{ij} \vert z\rangle \otimes \vert w\rangle =\vert \tfrac{1}{\sqrt{2}} (z-w)\rangle \otimes \vert \tfrac{1}{\sqrt{2}} (z+w)\rangle
$$
Here the subscripts refer to the primary and secondary modes, that is, the first and second components in the tensor product. In Figure \ref{eightport} the dashed side of the beam splitter represents the primary input mode. The phase shifter with phase shift $\theta$ is modelled by the unitary operator $R(\theta)=e^{i\theta N}$, where $N$ is the number operator.

Let $\rho$ and $\sigma $ be the states of modes 1 and 2, and let the local oscillator in mode 4 be in the coherent state $\vert \sqrt{2} z\rangle$. Mode 3 is assumed to be in the vacuum state $\vert 0 \rangle$. We detect the scaled number differences $\frac{1}{\sqrt{2}\vert z\vert}N_{13}^-$ and $\frac{1}{\sqrt{2}\vert z\vert}N_{24}^-$, where for example \cite{closure}
$$
\tfrac{1}{\sqrt{2}\vert z\vert}N_{13}^- =\tfrac{1}{\sqrt{2}\vert z\vert}( I_1 \otimes N_3 - N_1 \otimes I_3),
$$
and $N_j$ is the number operator related to the $j$th mode. The joint detection statistics are now described by the mapping \cite{biobservable}
$$
 (X,Y)\mapsto \E^{(\sqrt{2}\vert z\vert)^{-1}N_{13}^-} (X)\otimes \E^{(\sqrt{2}\vert z\vert)^{-1}N_{24}^-} (Y),
$$
where $\E^{(\sqrt{2}\vert z\vert)^{-1}N_{ij}^-}$ stands for the spectral measure of $\tfrac{1}{\sqrt{2}\vert z\vert}N_{ij}^-$.

\begin{figure}
\includegraphics[width=8cm]{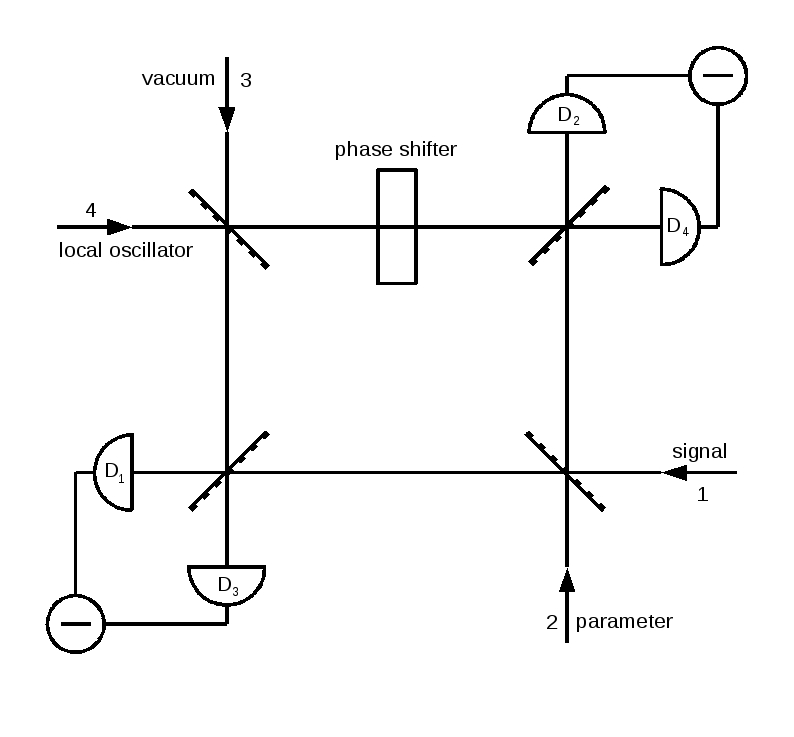}
\caption{Eight-port homodyne detector}\label{eightport}
\end{figure}

The high-amplitude limit $\vert z\vert\rightarrow \infty$ has been analyzed in detail in \cite{Kiukas2}. The signal observable $\E_{\theta,\sigma}$ measured in the high-amplitude limit is determined by the condition
\begin{eqnarray}
 \tr{\rho \E_{\theta, \sigma} (X\times Y) } &=&\textrm{tr}\big[ U_{12} (\rho\otimes \sigma ) U_{12}^* \nonumber\\
&&\times\sfq (\tfrac{1}{\sqrt{2}} X) \otimes \sfq_\theta(\tfrac{1}{\sqrt{2}}Y)\big]\label{highamplitude},
\end{eqnarray}
where $\sfq_\theta$ is the spectral measure of the quadrature operator \cite{closure} 
$Q_\theta =\frac{1}{\sqrt{2}} (e^{i\theta} a^* +e^{-i\theta} a )$ and $\sfq =\sfq_0$. With the choice $\theta =\tfrac{\pi}{2}$, this observable is covariant in the sense that
\begin{equation}\label{covariance}
 W(q,p) \E_{\frac{\pi}{2}, \sigma} (Z) W(q,p)^* =\E_{\frac{\pi}{2}, \sigma} (Z +(q,p)),
\end{equation}
where $W(q,p)= e^{i\frac{qp}{2}} e^{-iqP} e^{ipQ}$ is the Weyl operator. The general structure of covariant phase space observables is well-known; any  observable satisfying \eqref{covariance} is of the form \cite{Holevo, Werner} (for recent alternative proofs, see \cite{Kiukas, Cassinelli})
\begin{equation}\label{phasespaceobs}
 \G^S (Z) = \frac{1}{2\pi} \int_Z W(q,p) SW(q,p)^*\, dqdp
\end{equation}
for some unique positive trace one operator $S$, called the generating operator of $\G^S$. For $\theta=\frac{\pi}{2}$ the observable \eqref{highamplitude} is generated by the operator $C\sigma C^{-1}$, where $C$ is the conjugation map;  $(C\psi)(x)= \overline{\psi(x)}$, that is, $\E_{\frac{\pi}{2},\sigma} =\G^{C\sigma C^{-1}}$. However, for an arbitrary phase shift, it is straightforward to verify that the observable $\E_{\theta,\sigma}$ is not in general covariant. Indeed, we may use the coordinate representation
$$
(U_{12} \Psi )( x,y) = \Psi(\tfrac{1}{\sqrt{2}}(x+y), \tfrac{1}{\sqrt{2}}(-x+y))   
$$
for the beam splitter to calculate 
\begin{eqnarray*}
&&\tr{\rho W(q,p)\E_{\theta, \sigma} (X\times Y) W(q,p)^* } \\
 &=& \textrm{tr}[\rho \E_{\theta, \sigma} (X\times Y+ (q,q\cos\theta +p\sin\theta))  ]
\end{eqnarray*}
which shows that $\E_{\theta,\sigma}$ is covariant if and only if $\theta =\frac{\pi}{2}$. This means that for $\theta\neq\frac{\pi}{2}$ the measured observable is no longer of the form  \eqref{phasespaceobs}. However, we can solve the structure of the observable by considering covariance with respect to a different representation which we now construct explicitly.

For each $\theta \in (-\pi, 0) \cup (0,\pi)$ define a linear bijection $f_\theta:\R^2\rightarrow\R^2$  via 
$$
f_\theta(q,p) = (q,q\cos\theta +p\sin\theta).
$$
and define the tilted Weyl operators by $W_\theta (q,p) =W(f^{-1}_\theta(q,p))$. By operating on finite linear combinations of coherent states, we can verify the explicit form of these operators:
\begin{equation}
W_\theta (q,p) =e^{\frac{i}{2} \frac{qp}{\sin\theta}} e^{-\frac{iq}{\sin\theta}Q_\theta} e^{\frac{ip}{\sin\theta} Q}
\end{equation}
It is straightforward to check that the mapping $(q,p)\mapsto W_\theta (q,p)$ is in fact an irreducible projective unitary representation of $\R^2$. In particular, the equality 
$$
 W_\theta (q+q',p+p') = e^{\frac{i}{2} \frac{(qp'-pq')}{\sin\theta}}  W_\theta(q,p) W_\theta(q',p')
$$
holds for all $(q,p),(q',p')\in\R^2$. Note that this representation is not unitarily equivalent to $W$ unless $\theta =\frac{\pi}{2}$. Indeed, since 
$$
\left[ \frac{1}{\sin\theta} Q, \frac{1}{\sin\theta}Q_\theta\right] = \frac{i}{\sin\theta}I,
$$
the generators $\frac{1}{\sin\theta} Q$ and  $\frac{1}{\sin\theta}Q_\theta$ do not form a Weyl pair, and thus are not unitarily equivalent to $Q$ and $P$, unless $\theta =\frac{\pi}{2}$.

Similar to the case of \eqref{covariance}, we may consider the covariance of an observable with respect to this representation. We say that an observable $\E$ is a {\em $\theta$-covariant phase space observable} if 
$$
 W_{\theta}(q,p)\E (Z) W_\theta (q,p)^* =\E (Z+(q,p))
$$
for all $Z$ and $(q,p)$. The structure of such observables can be completely determined. In fact, according to \cite[Theorem 3]{Kiukas}, these observables are precisely of the form
\begin{equation}\label{thetaobs}
 \G^S_\theta (Z) = \frac{1}{2\pi\vert \sin\theta\vert} \int_Z W_\theta(q,p) SW_\theta(q,p)^*\, dqdp
\end{equation}
for some generating operator $S$. We immediately notice that even though $\G^S_\theta$ is not covariant in the usual sense, it is a function of a covariant observable. In fact, the equation $\G^S_\theta (Z) = \G^S (f^{-1}_\theta (Z))$ is easily verified for any set $Z$. In particular, the statistics of $\G^S_\theta$ can be calculated from the statistics of $\G^S$, and vice versa. As demonstrated in Figure \ref{densities}, for a fixed initial state and generating operator, the difference in measurement outcome statistics between two different values of $\theta$ can be viewed geometrically as a tilting of one axis in the phase space. Furthermore, since $f_\theta$ is a homeomorphism, we can state that \emph{the observables $\G^S$ and $\G^S_\theta$ are informationally equivalent}, so that in particular, \emph{$\G^S_\theta$ is informationally complete if and only if the support of the function $(q,p)\mapsto \tr{SW(q,p)}$ is the whole $\R^2$}.\cite{Kiukas3}

\begin{figure}
\includegraphics[width=3.6cm]{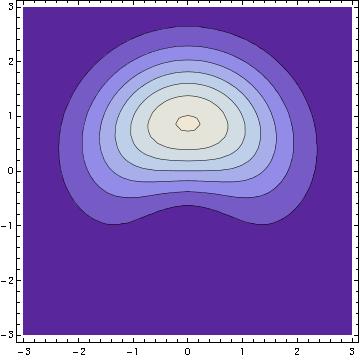}\qquad
\includegraphics[width=3.6cm]{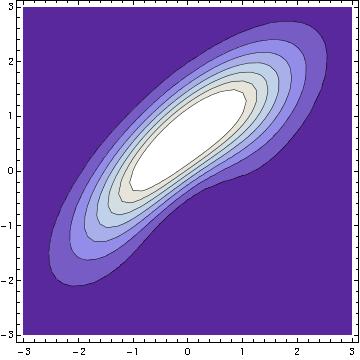}
\caption{Distributions corresponding to $\theta=\frac{\pi}{2}$ (left) and $\theta =\frac{\pi}{4}$ (right) for a fixed initial state and generating operator.}\label{densities}
\end{figure}

The Cartesian margins of $\G^S_\theta$ are the smeared quadratures $\mu^S *\sfq$ and $\mu_\theta^S*\sfq_\theta$ where, for instance, 
$$
\mu^S*\sfq (X) =\int \mu^S(X-q) \, d\sfq(q),
$$
and the convolving measures are given by $\mu^S_\theta (X) =\tr{\Pi S\Pi^* \sfq_\theta(X)}$ and  $\mu^S=\mu^S_0$, where $\Pi$ is the parity operator. The measurement of $\G^S_\theta$ thus consitutes an approximate joint measurement of the sharp quadratures $\sfq$ and $\sfq_\theta$. Furthermore, the proof of \cite[Proposition 7]{Carmeli} can be carried out using the tilted Weyl operators and we can therefore state that \emph{any two smeared quadratures  $\nu_1*\sfq$ and $\nu_2*\sfq_\theta$ have a joint observable if and only if they are margins of a $\theta$-covariant phase space observable.} Finally, Werner's result \cite{Werner2} also holds for $\sfq$ and $\sfq_\theta$: \emph{to any approximate joint observable for $\sfq$ and $\sfq_\theta$ there exists a $\theta$-covariant one which serves as a better approximation.}

In view of the eight-port homodyne detection scheme, it is easily seen that the observable defined in equation  \eqref{highamplitude} is in fact a $\theta$-covariant observable, that is, $\E_{\theta, \sigma} =\G^S_\theta$ for some $S$. Thus, what is left is the determination of the generating operator $S$, or more specifically, its dependence on the state of the parameter field $\sigma$. This is done in the following proposition, which constitutes the main result of this paper. We start with a few definitions.

Define the unitary operator $A_\theta$ by its action in the coordinate representation
$$
 (A_\theta \psi) (x) = \sqrt{\tfrac{ \vert\sin\theta\vert}{1+\cos\theta}} \, \psi (\sqrt{\tfrac{1-\cos\theta}{1+\cos\theta}} x)
$$
and the unitary operators $V_\theta$ by 
$$
 V_\theta = \left\{ \begin{array}{ll}
                  R\left(\tfrac{1}{2}(\theta-\pi)\right) A_\theta R\left(\tfrac{\theta}{2}\right), & \textrm{if } \theta \in(0,\pi)\\
 R\left(\tfrac{1}{2}(\theta+\pi)\right) A_\theta R\left(\tfrac{\theta}{2}\right), & \textrm{if }\theta \in (-\pi,0).
                 \end{array}\right.
$$
For any state $\sigma$, denote 
\begin{equation}\label{generator}
S_\theta(\sigma) =V_\theta C\sigma C^{-1}V_\theta^*.
\end{equation}
Now we are ready to prove  the explicit form of the generating operator.
\begin{proposition}
 For any parameter state $\sigma$ and phase shift $\theta\notin\{ 0,\pi\}$, the measured high-amplitude limit observable is $\G_\theta^{S_\theta(\sigma)}$.
\end{proposition}
\begin{proof}
First of all, notice that $\tr{\rho \G^{S_\theta(\sigma)}_\theta (Z)} =\tr{S_\theta(\sigma) \G^\rho_\theta (-Z)}$ for all $Z$, so that in  particular, $\langle 0\vert \G^{S_\theta(\sigma)}_\theta  (Z) \vert 0\rangle =\tr{S_\theta(\sigma) \G^{\vert 0\rangle}_\theta (-Z)}$. Since the observable $\G^{\vert 0\rangle}_\theta$ is informationally complete, it follows that the measurement outcome statistics corresponding to the vacuum signal state is sufficient to determine the generating operator.

Let $\sigma=\vert \psi\rangle\langle \psi\vert,$ where $\psi$ is a finite linear combination of coherent states, $\psi =\sum_{n=1}^k c_n \vert z_n\rangle$. A direct calculation shows that
\begin{eqnarray*}
 &&\langle 0\vert \G_\theta^{S_\theta(\sigma)} (X\times Y)\vert 0\rangle \\
&=& \frac{1}{2\pi\vert\sin\theta\vert} \int_{X\times Y} \big\vert \langle 0\vert W_\theta (q,p) V_\theta C\psi\rangle\big\vert^2 \, dqdp\\
&=& \frac{1}{2\pi\vert\sin\theta\vert} \sum_{m,n=1}^k \overline{c_m} c_n \int_{X\times Y}  \langle 0\vert W_\theta (q,p) V_\theta C\vert z_m\rangle\\
&&\qquad\times \overline{\langle 0\vert W_\theta (q,p) V_\theta C \vert z_n\rangle} \, dqdp\\
&=& \sum_{m,n=1}^k \overline{c_m} c_n  \textrm{tr}[U_{12} (\vert 0\rangle\langle 0\vert \otimes \vert z_n \rangle\langle z_m\vert )U_{12}^* \\
&&\qquad\times\sfq (\tfrac{1}{\sqrt{2}} X) \otimes \sfq_\theta (\tfrac{1}{\sqrt{2}} Y)]\\
&=& \tr{U_{12} (\vert 0\rangle\langle 0\vert \otimes \vert \psi \rangle\langle \psi\vert )U_{12}^* \sfq (\tfrac{1}{\sqrt{2}} X) \otimes \sfq_\theta (\tfrac{1}{\sqrt{2}} Y)} 
\end{eqnarray*}
and since the linear combinations of coherent states are dense in $\hil$, it follows that this equation holds for an arbitrary vector state $\psi$.

Now let $\sigma $ be an arbitrary state. By using the convex decomposition $\sigma =\sum_{n=0}^\infty \lambda_n P_n$, where $P_n=\vert \psi_n\rangle\langle \psi_n\vert$, we find that 
\begin{eqnarray*}
&&\langle 0\vert \G_\theta^{S_\theta(\sigma)} (X\times Y)\vert 0\rangle = \sum_{n=0}^\infty \lambda_n  \langle 0\vert \G_\theta^{S_\theta(P_n)} (X\times Y)\vert 0\rangle  \\
&=& \sum_{n=0}^\infty \lambda_n\tr{U_{12}^* (\vert 0\rangle\langle 0\vert \otimes P_n )U_{12}^* \sfq (\tfrac{1}{\sqrt{2}} X) \otimes \sfq_\theta (\tfrac{1}{\sqrt{2}} Y)}\\
&=& \tr{U_{12}^* (\vert 0\rangle\langle 0\vert \otimes \sigma )U_{12}^* \sfq (\tfrac{1}{\sqrt{2}} X) \otimes \sfq_\theta (\tfrac{1}{\sqrt{2}} Y)},
\end{eqnarray*}
from which it follows that 
\begin{eqnarray*}
\tr{\rho\G_\theta^{S_\theta(\sigma)} (X\times Y)} &=&\textrm{tr}[U_{12}^* (\rho \otimes \sigma )U_{12}^* \\
&&\qquad\times\sfq (\tfrac{1}{\sqrt{2}} X) \otimes \sfq_\theta (\tfrac{1}{\sqrt{2}} Y)]
\end{eqnarray*}
for all $\rho$ and $\sigma$, which proves our claim.
\end{proof}

\begin{figure}
\includegraphics[width=3.6cm]{Pic1.jpg}\qquad
\includegraphics[width=3.6cm]{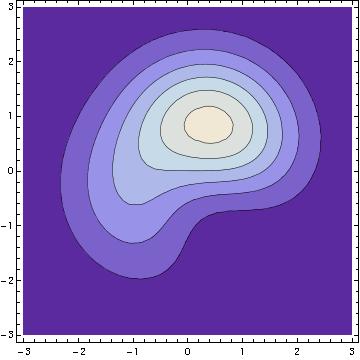}
\caption{Distributions corresponding to $\theta=\frac{\pi}{2}$ (left) and $\theta =\frac{\pi}{4}$ (right) for a fixed initial state and parameter field. Note that the generating operator changes according to \eqref{generator}.}\label{measured}
\end{figure}

As a final observation, we wish to take into account the non-unit quantum efficiencies of the photon detectors. A detailed analysis in the case of $\theta=\frac{\pi}{2}$ has been done in \cite{Lahti}. Suppose that each detector $D_j$ is assigned with a quantum efficiency $\epsilon_j \in (0,1)$, so that each detector constitutes a measurement of the approximate number observable

$$
n\mapsto \sum_{m=n}^\infty \binom{m}{n} \epsilon_j^n (1-\epsilon_j)^{m-n} \vert m\rangle\langle m\vert. 
$$
Then the spectral measures $\sfq$ and $\sfq_\theta$ in the defining equation \eqref{highamplitude} are replaced with their smearings $\mu_{13}*\sfq$ and $\mu_{24}*\sfq_\theta$, where for instance, $\mu_{13}$ is a probability measure having a Gaussian density $\sqrt{\tfrac{ 2\epsilon_1\epsilon_3}{\pi(\epsilon_1 -2\epsilon_1\epsilon_3 +\epsilon_3)}} e^{-\tfrac{ 2\epsilon_1\epsilon_3}{\epsilon_1 -2\epsilon_1\epsilon_3 +\epsilon_3} x^2}$. Then, by defining the phase space probability measure $\mu_\epsilon (X\times Y) =\mu_{13} (\tfrac{1}{\sqrt{2}} X) \mu_{24}(\tfrac{1}{\sqrt{2}} Y)$, we can show that the measured observable is the smeared $\theta$-covariant phase space observable $\mu_\epsilon* \G^{S(\sigma)}$. This smearing is also covariant with respect to the same representation and thus there exists a generating operator $S_\theta'(\sigma)$ such that $\mu_\epsilon *\G^{S_\theta(\sigma)} =\G^{S_\theta'(\sigma)}$. Just as in the case $\theta =\frac{\pi}{2}$ treated in \cite{Lahti}, this generating operator can be expressed as a convolution $S'_\theta(\sigma) =\mu_\epsilon *S_\theta(\sigma)$, where \cite{Werner}
$$
\mu_\epsilon *S_\theta(\sigma) =\int W_\theta (q,p)S_\theta(\sigma) W_\theta(q,p)^*\, d\mu_\epsilon (q,p).
$$
From the tomographic point of view, the presence of detector inefficiences is merely an inconvenience since, as shown in \cite{Lahti}, the smeared observable is informationally equivalent to the one corresponding to ideal detectors.

\textbf{Conclusions.} We have considered the eight-port homodyne detection scheme in the case of an arbitrary phase shift $\theta$. We have shown that to each phase shift there is a corresponding projective representation such that the measured high-amplitude limit observable is covariant with respect to it. We have interpreted this geometrically as a tilting of one axis in the phase space of the system. Furthermore, we have constructed explicitly the measured observable and considered some of its properties. In particular, we have shown that this measurement scheme can be used for quantum tomography and approximate joint measurements of an arbitrary pair of rotated quadratures.

\noindent
\textbf{Acknowledgment.} J. S. was supported by the Finnish Cultural Foundation during the research leading to this manuscript.


\begin{thebibliography}{99}

\expandafter\ifx\csname natexlab\endcsname\relax\def\natexlab#1{#1}\fi
\expandafter\ifx\csname bibnamefont\endcsname\relax
  \def\bibnamefont#1{#1}\fi
\expandafter\ifx\csname bibfnamefont\endcsname\relax
  \def\bibfnamefont#1{#1}\fi
\expandafter\ifx\csname citenamefont\endcsname\relax
  \def\citenamefont#1{#1}\fi
\expandafter\ifx\csname url\endcsname\relax
  \def\url#1{\texttt{#1}}\fi
\expandafter\ifx\csname urlprefix\endcsname\relax\def\urlprefix{URL }\fi
\providecommand{\bibinfo}[2]{#2}
\providecommand{\eprint}[2][]{\url{#2}}


\bibitem[{\citenamefont{Werner}(1984)}]{Werner2}
\bibinfo{author}{\bibfnamefont{R.}~\bibnamefont{Werner}},
    \bibinfo{journal}{Quantum Inf. Comput.} \textbf{\bibinfo{volume}{4}},
  \bibinfo{pages}{546} (\bibinfo{year}{2004}).



\bibitem[{\citenamefont{Ali et~al.}(2005)\citenamefont{Ali, and Prugove\v cki}}]{Ali}
\bibinfo{author}{\bibfnamefont{S. T.}~\bibnamefont{Ali}},
  \bibnamefont{and}
  \bibinfo{author}{\bibfnamefont{E.}~\bibnamefont{Prugove\v cki}},
  \bibinfo{journal}{Physica} \textbf{\bibinfo{volume}{89A}},
  \bibinfo{pages}{501} (\bibinfo{year}{1977}).



\bibitem[{\citenamefont{Prugovecki}(1977)}]{Prugovecki}
\bibinfo{author}{\bibfnamefont{E.}~\bibnamefont{Prugovecki}},
    \bibinfo{journal}{Int. J. Theor. Phys.} \textbf{\bibinfo{volume}{16}},
  \bibinfo{pages}{321} (\bibinfo{year}{1977}).



\bibitem[{\citenamefont{Kiukas et~al.}(2008)\citenamefont{Kiukas, and Lahti}}]{Kiukas2}
\bibinfo{title}{Heuristic arguments have long been known. A rigorous quantum mechanical proof is given in}
\bibinfo{author}{\bibfnamefont{J.}~\bibnamefont{Kiukas}},
   \bibnamefont{and}
  \bibinfo{author}{\bibfnamefont{P.}~\bibnamefont{Lahti}},
  \bibinfo{journal}{J. Mod. Opt.} \textbf{\bibinfo{volume}{55}},
  \bibinfo{pages}{1891} (\bibinfo{year}{2008}).


\bibitem{closure}
\bibinfo{title}{To be more precise, one needs to take the closure of this essentially self-adjoint operator.}



\bibitem[{\citenamefont{Lahti et~al.}(1998)\citenamefont{Lahti, {Pulmannova}, and Ylinen}}]{biobservable}
\bibinfo{title}{Technically, this mapping is a biobservable, see for instance,}
\bibinfo{author}{\bibfnamefont{P.}~\bibnamefont{Lahti}},
  \bibinfo{author}{\bibfnamefont{S.}~\bibnamefont{{Pulmannova}}},
  \bibnamefont{and}
  \bibinfo{author}{\bibfnamefont{K.}~\bibnamefont{Ylinen}},
  \bibinfo{journal}{J. Math. Phys.} \textbf{\bibinfo{volume}{39}},
  \bibinfo{pages}{6364} (\bibinfo{year}{1998}).




\bibitem[{\citenamefont{Holevo}(1979)}]{Holevo}
\bibinfo{author}{\bibfnamefont{A. S.}~\bibnamefont{Holevo}},
    \bibinfo{journal}{Rep. Math. Phys.} \textbf{\bibinfo{volume}{16}},
  \bibinfo{pages}{385} (\bibinfo{year}{1979}).

\bibitem[{\citenamefont{Werner}(1984)}]{Werner}
\bibinfo{author}{\bibfnamefont{R.}~\bibnamefont{Werner}},
    \bibinfo{journal}{J. Math. Phys.} \textbf{\bibinfo{volume}{25}},
  \bibinfo{pages}{1404} (\bibinfo{year}{1984}).

\bibitem[{\citenamefont{Kiukas et~al.}(2006)\citenamefont{Kiukas, {Lahti}, and Ylinen}}]{Kiukas}
\bibinfo{author}{\bibfnamefont{J.}~\bibnamefont{Kiukas}},
  \bibinfo{author}{\bibfnamefont{P.}~\bibnamefont{{Lahti}}},
  \bibnamefont{and}
  \bibinfo{author}{\bibfnamefont{K.}~\bibnamefont{Ylinen}},
  \bibinfo{journal}{J. Math. Anal. Appl.} \textbf{\bibinfo{volume}{319}},
  \bibinfo{pages}{783} (\bibinfo{year}{2006}).


\bibitem[{\citenamefont{Cassinelli et~al.}(2003)\citenamefont{Cassinelli, {De Vito}, and Toigo}}]{Cassinelli}
\bibinfo{author}{\bibfnamefont{G.}~\bibnamefont{Cassinelli}},
  \bibinfo{author}{\bibfnamefont{E.}~\bibnamefont{{De Vito}}},
  \bibnamefont{and}
  \bibinfo{author}{\bibfnamefont{A.}~\bibnamefont{Toigo}},
  \bibinfo{journal}{J. Math. Phys.} \textbf{\bibinfo{volume}{44}},
  \bibinfo{pages}{4768} (\bibinfo{year}{2003}).


\bibitem{Kiukas3}
\bibinfo{title}{J. Kiukas and R. Werner, private communication}.

\bibitem[{\citenamefont{Carmeli et~al.}(2005)\citenamefont{Carmeli, {Heinonen}, and Toigo}}]{Carmeli}
\bibinfo{author}{\bibfnamefont{C.}~\bibnamefont{Carmeli}},
  \bibinfo{author}{\bibfnamefont{T.}~\bibnamefont{{Heinonen}}},
  \bibnamefont{and}
  \bibinfo{author}{\bibfnamefont{A.}~\bibnamefont{Toigo}},
  \bibinfo{journal}{J. Phys. A: Math. Gen.} \textbf{\bibinfo{volume}{38}},
  \bibinfo{pages}{5253} (\bibinfo{year}{2005}).





\bibitem[{\citenamefont{Lahti et~al.}(2010)\citenamefont{Lahti, {Pellonp\"a\"a}, and Schultz}}]{Lahti}
\bibinfo{author}{\bibfnamefont{P.}~\bibnamefont{Lahti}},
  \bibinfo{author}{\bibfnamefont{J.-P.}~\bibnamefont{{Pellonp\"a\"a}}},
  \bibnamefont{and}
  \bibinfo{author}{\bibfnamefont{J.}~\bibnamefont{Schultz}},
  \bibinfo{journal}{J. Mod. Opt.} \textbf{\bibinfo{volume}{57}},
  \bibinfo{pages}{1171} (\bibinfo{year}{2010}).






\end{thebibliography}
\end{document}